\documentclass[conference]{IEEEtran}
\IEEEoverridecommandlockouts
\usepackage{cite}
\usepackage{enumitem}
\usepackage{fancyhdr,setspace}
\setlist[itemize]{noitemsep, topsep=0pt, leftmargin=12pt}
\usepackage[font=small,skip=2pt]{caption}
\usepackage{booktabs,caption}
\usepackage[flushleft]{threeparttable}
\usepackage{amsmath,amssymb,amsfonts,threeparttable}
\usepackage[linesnumbered,boxed]{algorithm2e}
\usepackage[noend]{algpseudocode}
\usepackage[compact]{titlesec}
\usepackage[T1]{fontenc}
\usepackage{times} 
\usepackage{amsmath}

\algnewcommand{\IfThenElse}[3]{
  \State \algorithmicif\ #1\ \algorithmicthen\ #2\ \algorithmicelse\ #3}

\SetKwInput{KwInput}{Input}                
\SetKwInput{KwOutput}{Output} 
\makeatletter
\newcommand{\removelatexerror}{\let\@latex@error\@gobble}
\makeatother

\newcommand{\tablefontsize}{\fontsize{8.5pt}{8pt}\selectfont}

\usepackage{graphicx}
\usepackage{float}
\usepackage{subfig}
\usepackage{textcomp}
\usepackage{xcolor}
\usepackage{multirow}
\usepackage{amsthm}

\newtheorem{theorem}{Theorem}

\def\BibTeX{{\rm B\kern-.05em{\sc i\kern-.025em b}\kern-.08em
    T\kern-.1667em\lower.7ex\hbox{E}\kern-.125emX}}
    
\begin{document}

\IEEEoverridecommandlockouts
\IEEEpubid{\makebox[\columnwidth]{ 979-8-3315-2144-8/25/\textdollar31.00 \copyright2025 IEEE \hfill} \hspace{\columnsep}\makebox[\columnwidth]{1}}

\makeatletter
\renewcommand \footnoterule{
  \kern-3\p@
  \hrule\columnwidth
  \kern2.6\p@}
  
\makeatother
\title{Defect Analysis and Built-In-Self-Test for Chiplet Interconnects in Fan-out Wafer-Level Packaging$^*$
}
\author{    
  \IEEEauthorblockN{Partho Bhoumik, Christopher Bailey, and Krishnendu Chakrabarty}    \IEEEauthorblockA{School of Electrical, Computer and Energy Engineering, Arizona State University \\ ASU Center for Semiconductor Microelectronics (ACME), Arizona State University}  
  }
\maketitle

\begin{abstract}
Fan-out wafer-level packaging (FOWLP) addresses the demand for higher interconnect densities by offering reduced form factor, improved signal integrity, and enhanced performance. However, FOWLP faces manufacturing challenges such as coefficient of thermal expansion (CTE) mismatch, warpage, die shift, and post-molding protrusion, causing misalignment and bonding issues during redistribution layer (RDL) buildup. Moreover, the organic nature of the package exposes it to severe thermo-mechanical stresses during fabrication and operation. In order to address these challenges, we propose a comprehensive defect analysis and testing framework for FOWLP interconnects. We use Ansys Q3D to map defects to equivalent electrical circuit models and perform fault simulations to investigate the impacts of these defects on chiplet functionality. Additionally, we present a built-in self-test (BIST) architecture to detect stuck-at and bridging faults while accurately diagnosing the fault type and location. Our simulation results demonstrate the efficacy of the proposed BIST solution and provide critical insights for optimizing design decisions in packages, balancing fault detection and diagnosis with the cost of testability insertion.

\end{abstract}

\section{Introduction} 

Advanced packaging has emerged as a pivotal solution for extending the capabilities of integrated circuits (ICs) \cite{lau2022recent}. 
Traditionally, semiconductor devices were designed as large monolithic dies, but as circuits became more complex, industry shifted toward dividing these large dies into smaller, specialized chiplets with higher yields. These chiplets can be interconnected to enable heterogeneous integration, facilitating the development of advanced 2.5D and 3D packages \cite{li2024high, chen2024survey}.

Among the various advanced packaging techniques, fan-out wafer-level packaging (FOWLP) offers several advantages over conventional methods, including enhanced electrical performance, reduced form factor, and higher I/O density \cite{hudson2021deca, lee2017fowlp, liu2012high}. This approach involves using a secondary wafer to position the dies with more separation, thereby allowing the extra area to develop redistribution layers (RDL) that connect the neighboring dies via Cu pillars grown on the die pads. However, coefficient of thermal expansion (CTE) mismatch, warpage, and misalignment during bonding can lead to defects that compromise the integrity of the interconnects \cite{wang2021warpage, ding2016molding, partho_warp, parthowarp2, parthowarp3, Lim_protrusion}. A comprehensive study of these defects is essential to understand the impact of defect size and develop test solutions.

Our goal is to tackle the problem of modeling and assessing the impact of various defects within Cu pillars and RDL in a FOWLP package. These defects can originate both during packaging unit processes (such as the RDL build-up processes) and in-service operations due to thermo-mechanical loads \cite{lau2018reliability, 9696476}. Our approach leverages finite element analysis (FEA) to extract the parasitics corresponding to various defect levels, utilizing the Ansys Q3D tool. We show how the impact of these defects can be analyzed by mapping them to their equivalent faulty circuits.
\par
Furthermore, built-in self-test (BIST) solutions are essential for advanced packaging \cite{hir}. BIST can be used to detect and locate opens and shorts within interconnects. Once faults are detected and localized, rerouting strategies can be employed to repair the faulty connections. Despite its importance, there has been limited research on package-level BIST. In \cite{cui2023physical}, a physically aware test method was proposed for interconnects. This method utilizes an array of pattern generators to apply patterns through the chiplet I/Os. However, it does not define fault models or specify test patterns for these models. It also lacks a fault detection mechanism and a diagnostic scheme for failure analysis. The authors in \cite{chuang2024generating, wang2024test} investigated test pattern generation schemes and proposed a test method using automated test equipment and die-wrapper registers to apply the patterns. A total of 16 patterns (4 unique codewords) can test for hard and weak defects within chiplet interconnects. However, scanning test patterns in and out of the tester for thousands of I/Os takes a prohibitive amount of time. An on-chip solution that addresses both post-assembly and in-field fault scenarios is essential for lifecycle monitoring.
\newline
In this paper, we tackle the challenges associated with detecting and localizing defects within chiplet interconnects. We begin by analyzing these defects, which allows us to identify and quantify their physical and electrical characteristics. This detailed analysis also helps us to understand the relationship between defects and fault models. Based on these fault models, we propose an efficient approach involving only three test patterns to detect, diagnose and localize all interconnect stuck-at-faults (SAFs) and bridging faults. 
The main contributions of this paper are listed as follows:

\begin{itemize}

\item We provide a comprehensive characterization of defects inside Cu pillars and RDL.
\item \normalfont {We demonstrate how to map defects in package components to their equivalent faulty circuits and perform package-level fault simulation using HSPICE.} 

\item \normalfont {We present a BIST architecture that detects, locates, and diagnoses stuck-at (SA) and bridging faults. The BIST outcomes can be combined with defect simulation results to identify the extent of deformities in Cu pillars and RDLs.} 
\end{itemize}

The rest of the paper is organized as follows. Section~\ref{sec:background} provides background on FOWLP and FEA. Section~\ref{sec:defect modeling} describes the proposed defect analysis framework and equivalent faulty circuits. In Section~\ref{sec: bist}, the BIST architecture is presented and evaluated, and BIST outcomes are mapped to deformities in package components. Section~\ref{sec:conclusion} concludes the paper.

\section{Background}\label{sec:background}
\subsection{Fan-out Wafer Level Packaging}
Conventional packaging technologies, such as flip-chip, FOVEROS, EMIB, and CoWoS, rely heavily on substrates or interposers \cite{sheikh20212}, which introduce additional layers, cost, and complexity. 
FOWLP is an alternative approach that encapsulates the entire die within a molding compound, creating a smooth and continuous surface. 
In a chip-first, face-up FOWLP method, Cu pillars are first grown at the I/O locations on the native semiconductor wafer. After that, the dies are placed on a reconstituted wafer panel using the die pick and place method. During panelization, the dies are covered on both the front and sides with epoxy molding compound (EMC), with the Cu pillars providing the current pathways through the mold. The panel is then processed by building up interconnections between the exposed Cu pillars and the Cu RDL. An epoxy laminate is applied to the backside of the panel to fully encase the die in epoxy; see Fig.~\ref{fig:deca_package}. Adaptive Patterning (AP) is used to dynamically scan the position of the Cu bumps and adjust interconnect routing to reduce the impact of die shift \cite{san2023integrating}.


\begin{figure}[!tbp]
    \centering
    \includegraphics[width=0.6\columnwidth]{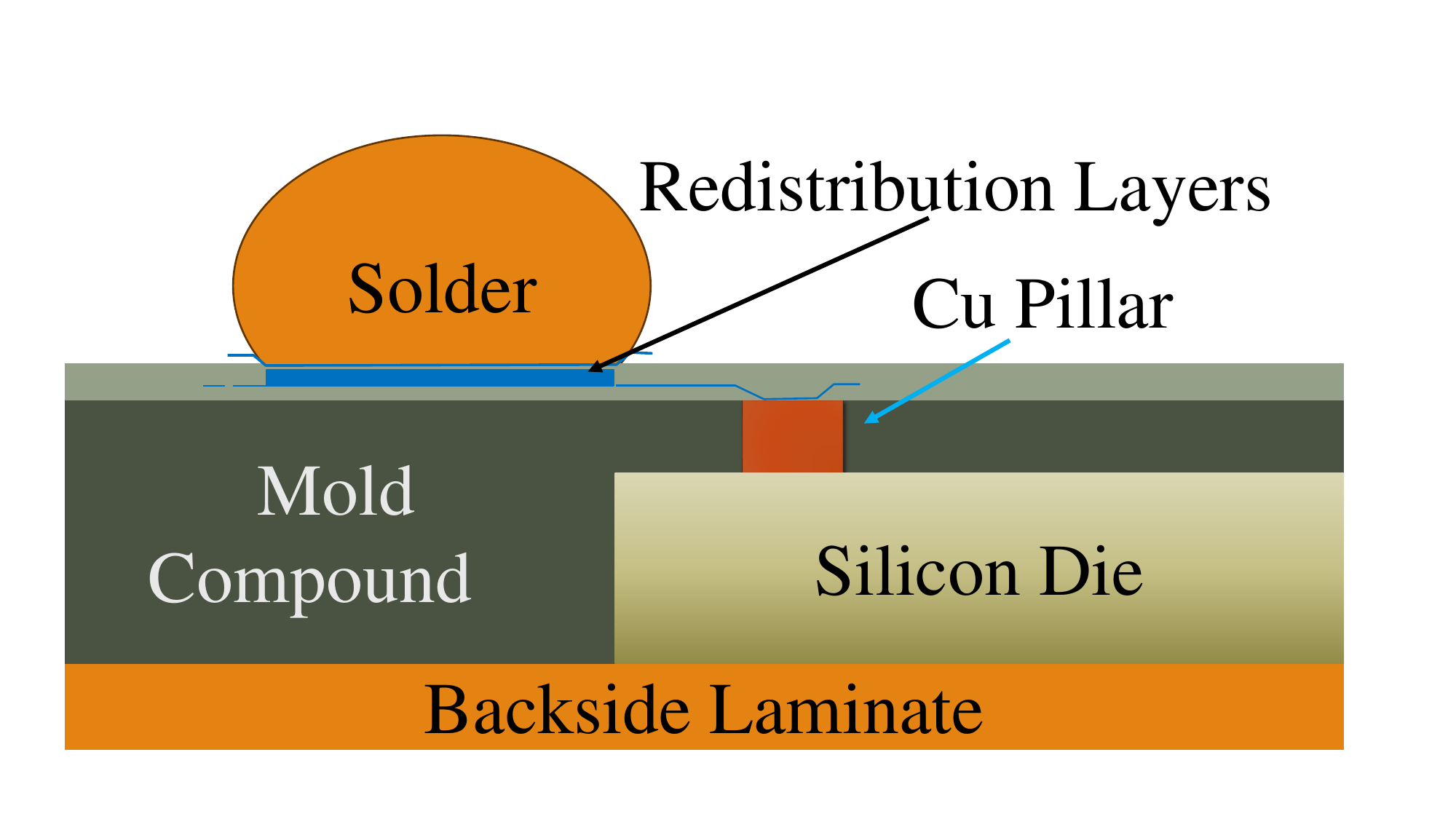}
    \caption{Illustration of a fully molded FOWLP structure \cite{rogers2013implementation}.}
    \label{fig:deca_package}
\vspace{-1.4em}
\end{figure}

Testability solutions for this promising technology have yet to be studied. \cite{rogers2013implementation} demonstrated a test vehicle approach to assess interconnect continuity using a daisy chain. This approach is limited to test vehicles and does not extend to practical applications. Ansys modeling was carried out to simulate the thermal cycling performance of the test package, focusing on strain energy density distribution in critical solder joints. However, there are many other sources of defects that need to be studied. Comprehensive design for test structures is necessary to detect, diagnose, and repair defects.

\subsection{Finite Element Analysis}
FEA begins by creating a mesh of the object, applying appropriate boundary conditions and loads, and solving mathematical equations describing each element's physical behavior. While traditional FEA is used for mechanical and thermal analysis, tools such as Ansys Q3D extend these principles to the electromagnetic domain, enabling precise extraction of electrical parasitics. Ansys Q3D Extractor is a 3D quasi-static electromagnetic field simulator that directly extracts parasitic parameters (R, L, C) for complex 3D structures. 

\label{subsec:fea}

\section{Defect Analysis}\label{sec:defect modeling}
\subsection{ Defects and Equivalent Circuits}

{\bf Cu Pillars:}  These vertical elements extend from the Si die through the molding compound and connect to the RDL. The CTE mismatch between copper (CTE six times higher than silicon \cite{6231431}) and the other materials can induce significant thermal and mechanical stress. This stress may result in cracks or fractures within the Cu pillars. Moreover, adaptive patterning utilized in FOWLP is highly sensitive, and minor deviations during dynamic scanning can lead to increased contact resistance between Cu bumps and RDL \cite{san2023integrating}. Furthermore, improper warpage management can cause small-pitch Cu pillars (20 µm) to be shorted to each other.

{\bf RDL:} The RDL pitch can be as small as 2 µm. At such fine pitch, residue and dust particles can lead to bridging and increased coupling between adjacent lines \cite{vangai2017comprehensive}. Additionally, incomplete or over-etching during manufacturing can damage RDL layers. The above defects are summarized in Table~\ref{tab:defect_list}. 
\begin{table}[h!]
\centering
\tablefontsize
\caption{Summary of defects and root causes.}
\renewcommand{\arraystretch}{1.2} 
\begin{tabular}{|p{1.4cm}|p{1.8cm}|p{4cm}|}
\hline
\textbf{Component}        & \textbf{Defect Type} & \textbf{Root Cause}                                     \\ \hline
\multirow{3}{*}{Cu Pillars} & Cracks/Breaks       & CTE mismatch                                              \\  
                            & Misalignment       & AP deviation or Cu smearing                               \\  
                            & Bridge             & Warpage-induced merging                                   \\ \hline
\multirow{3}{*}{RDL}        & Coupling          & Dielectric-conductor proximity                            \\  
                            & Bridge             & Metal residue during fabrication                          \\  
                            & Damaged RDL        & Etching errors, electromigration                          \\ \hline
\end{tabular}
\label{tab:defect_list}
\vspace{-1.3em}
\end{table}

\subsection{Fault-free RLC Extraction and Equivalent Circuit}
Before modeling defects, it is essential to first determine the RLC parameters of the interconnect components in the nominal case. We created the geometry of these three elements using Q3D Extractor. The dimensions and material of the elements, along with the extracted parasitics, are presented in Table ~\ref{tab:RLC}. Note that the RDL length will vary depending on the placement of chiplets. To account for this variability, we provide an equation derived from our simulation to calculate the self and mutual capacitance of RDL for any length $L$.

\begin{table}[h!]
\centering
\tablefontsize
\caption {Details about the interconnect components.}
\renewcommand{\arraystretch}{1.3}
\begin{tabular}{|p{2.5cm}|p{1.6cm}|p{2.5cm}|}
\hline
\textbf{Property}                             & \textbf{Cu Pillar}                        & \textbf{RDL}                                     \\ \hline
Dimensions              & d: 20 µm \newline h: 20 µm      & $L$ µm $\times$ 2 µm $\times$ 2 µm \\ \hline
Resistance                                    & 1.11 m\(\Omega\)                                    & 4.31 m\(\Omega\)/µm                                           \\ \hline
Self Capacitance                                  & 3.21 fF                                     & $(1+(L/5-1)\times 0.72)\times 0.7$ fF                                                  \\ \hline
Mutual  Capacitance                           & Negligible                           & $0.092 \times L$ fF                                     \\ \hline
\end{tabular}
\label{tab:RLC}
\vspace{-0.9em}
\end{table}



The resistance values for all components closely align with the hand-calculated results, whereas the capacitance values are higher than expected. The discrepancy arises because the simplification for hand calculation does not account for fringe effects at the edges of the device. In reality, fringe effects can contribute up to 50\% of the total capacitance at this dimension \cite{ibm}, resulting in higher overall values than calculated. Nevertheless, we can conclude that Q3D accurately models these effects, offering realistic and comprehensive capacitance parasitics. The lumped simulation model of a Cu pillar and RDL segment is illustrated in Fig.~\ref{fig:equivalent_circuits}.

\subsection{Ansys Modeling and Equivalent Faulty Circuits}

{\bf Cu Pillars:} We consider a crack at a height of 10 µm within the Cu pillar, characterized by a thickness of 0.5 nm. This defect is introduced in such a manner that it does not extend across the entire cross-sectional area of the pillar, thereby preserving a residual conductive path, as depicted in Fig.~\ref{fig:Ansys_figure}. Based on the remaining conductive cross-section, we extract the resistance $R_f$ and capacitance $C_f$ values to develop the equivalent faulty circuit associated with this defect (Fig.~\ref{fig:equivalent_circuits}). In cases where the pillar is completely severed with no conductive path remaining, a capacitance is formed between the top and bottom segments of the pillar. For this scenario, only the $C_f$  value needs to be considered.

\begin{figure}[!tbp]
    \centering
    \includegraphics[width=1\columnwidth]{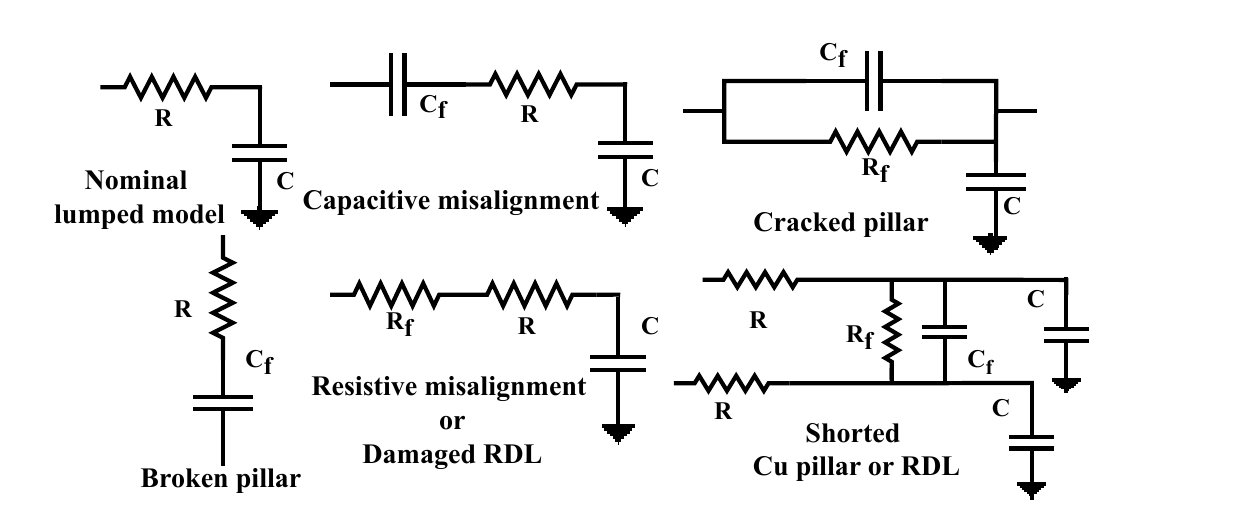}
    \caption{Equivalent nominal and faulty circuits for Cu pillars and RDL.}
    \label{fig:equivalent_circuits}
\vspace{-1.9em}
\end{figure}

Misalignment with the RDL can manifest in two ways:

1. Resistive Misalignment: This defect occurs when the RDL maintains contact with the pillar, but the overlap between the two components decreases due to die shift. The reduction in surface contact area increases resistance, which we model as $R_f$. We add a cylindrical base atop the pillar to serve as the RDL contact point. By varying the area of this cylindrical base, we extract the misalignment-induced resistance for various defect severities. Note that incomplete bonding due to pressure and thermal effects may introduce additional contact resistance, which must be incorporated into the extracted value as the Q3D tool does not account for these extrinsic factors.

2. Capacitive Misalignment: This occurs when the RDL fails to make proper contact with the Cu bump (Fig.~\ref{fig:Ansys_figure}), leaving a gap that forms a capacitive defect modeled by $C_f$ in Fig.~\ref{fig:equivalent_circuits}.

To simulate bridging defects, we introduce metal residues modeled as a combination of spherical and cylindrical geometries that span between adjacent bumps (Fig.~\ref{fig:Ansys_figure}). The radius of this conductive bridge is varied to account for different levels of defect severity, with Q3D being used to extract the corresponding resistance $R_f$.
\begin{figure}[!tbp]
    \centering
    \includegraphics[width=0.8\columnwidth]{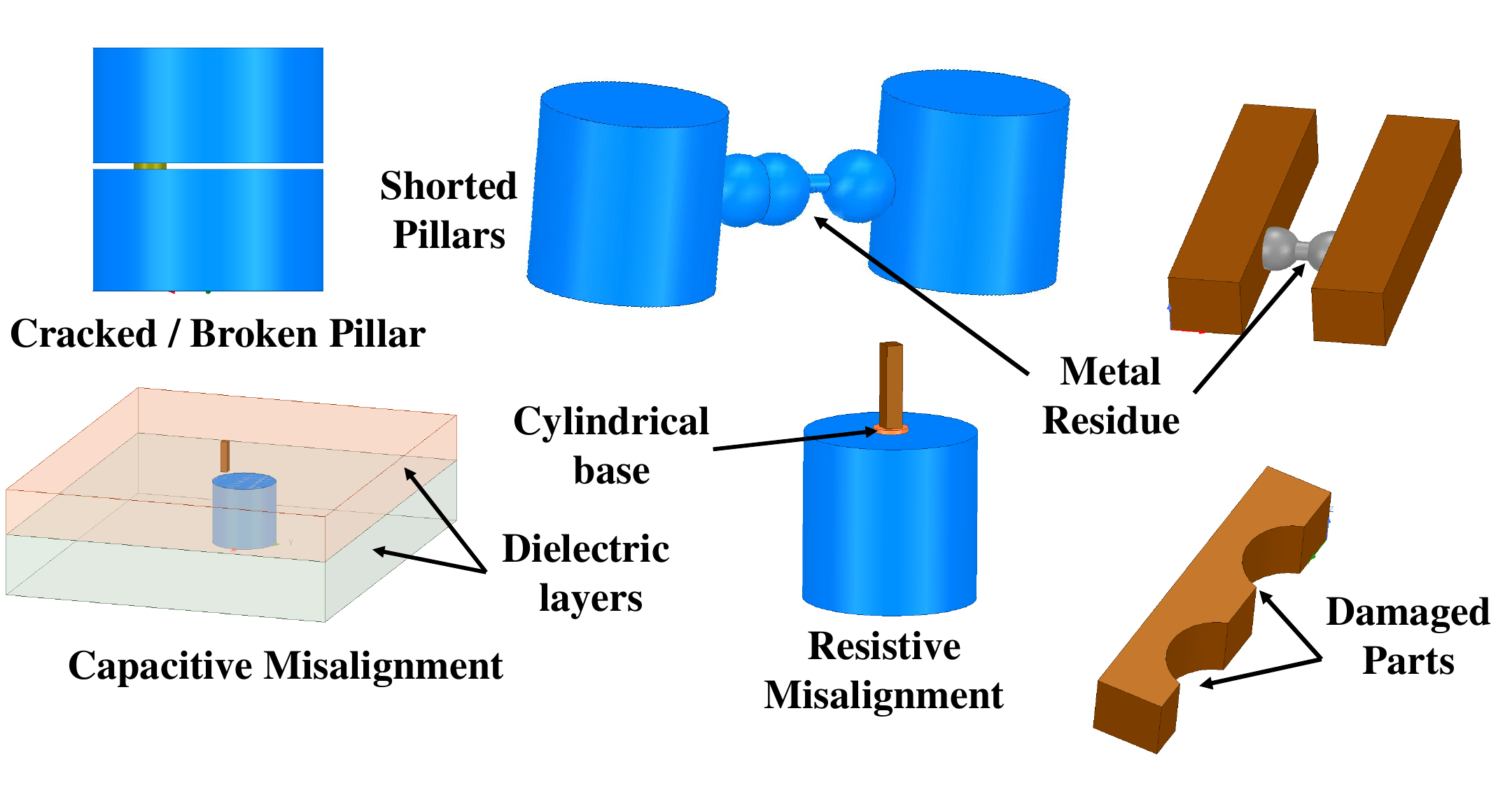}
    \caption{Ansys modeling of Cu pillars and RDL.}
    \label{fig:Ansys_figure}
    \vspace{-2em}
\end{figure}

{\textbf{RDL:}} The analysis of RDL bridging defects mirrors that of the Cu pillar. However, the finer pitch (2 µm) of the RDL results in a shorter short-circuit path, thereby amplifying the impact of the defect when compared to Cu bumps. Additionally, the coupling capacitance between adjacent RDL layers, which is neglected for Cu bumps due to their larger pitch (20 µm), must be considered here. The metal residue between two adjacent RDL layers is shown in Fig.~\ref{fig:Ansys_figure}. Two 10 µm RDL sections are considered to compute the $R_f$ for the short.

\begin{figure*}[!t]
    \centering
    
    \subfloat[]{\includegraphics[width=0.28\textwidth]{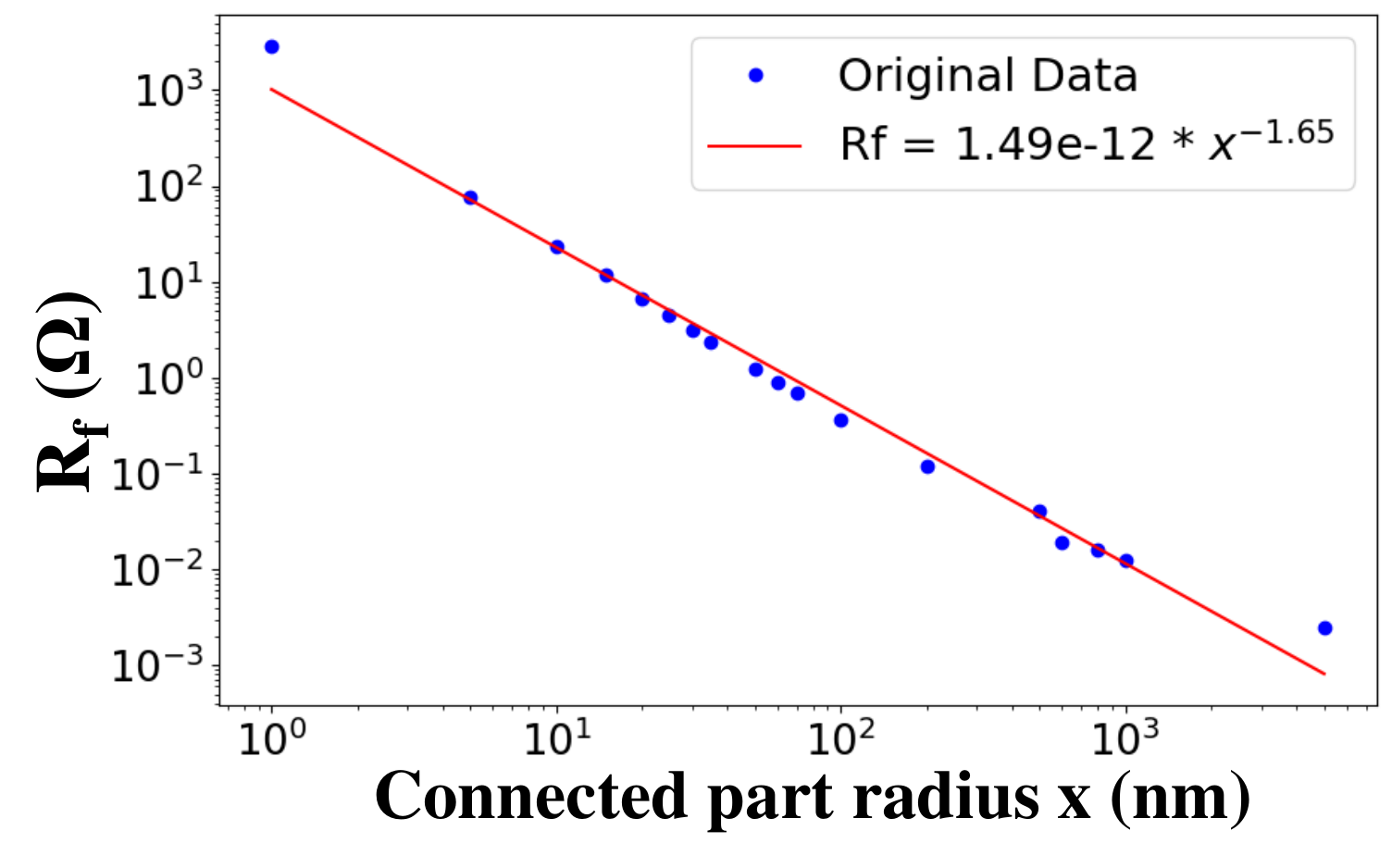}} 
    \subfloat[]{\includegraphics[width=0.28\textwidth]{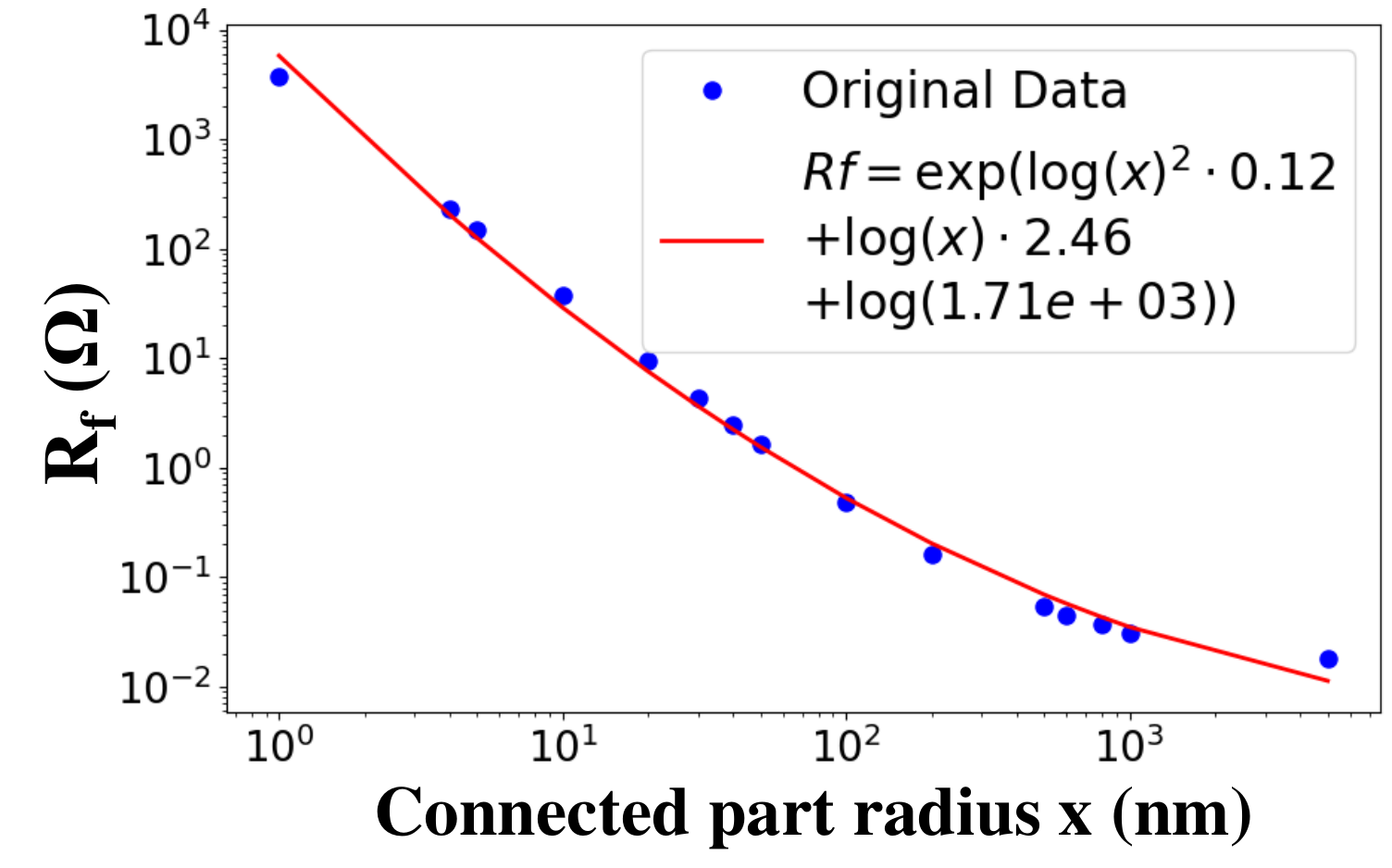}}
    \subfloat[]{\includegraphics[width=0.28\textwidth]{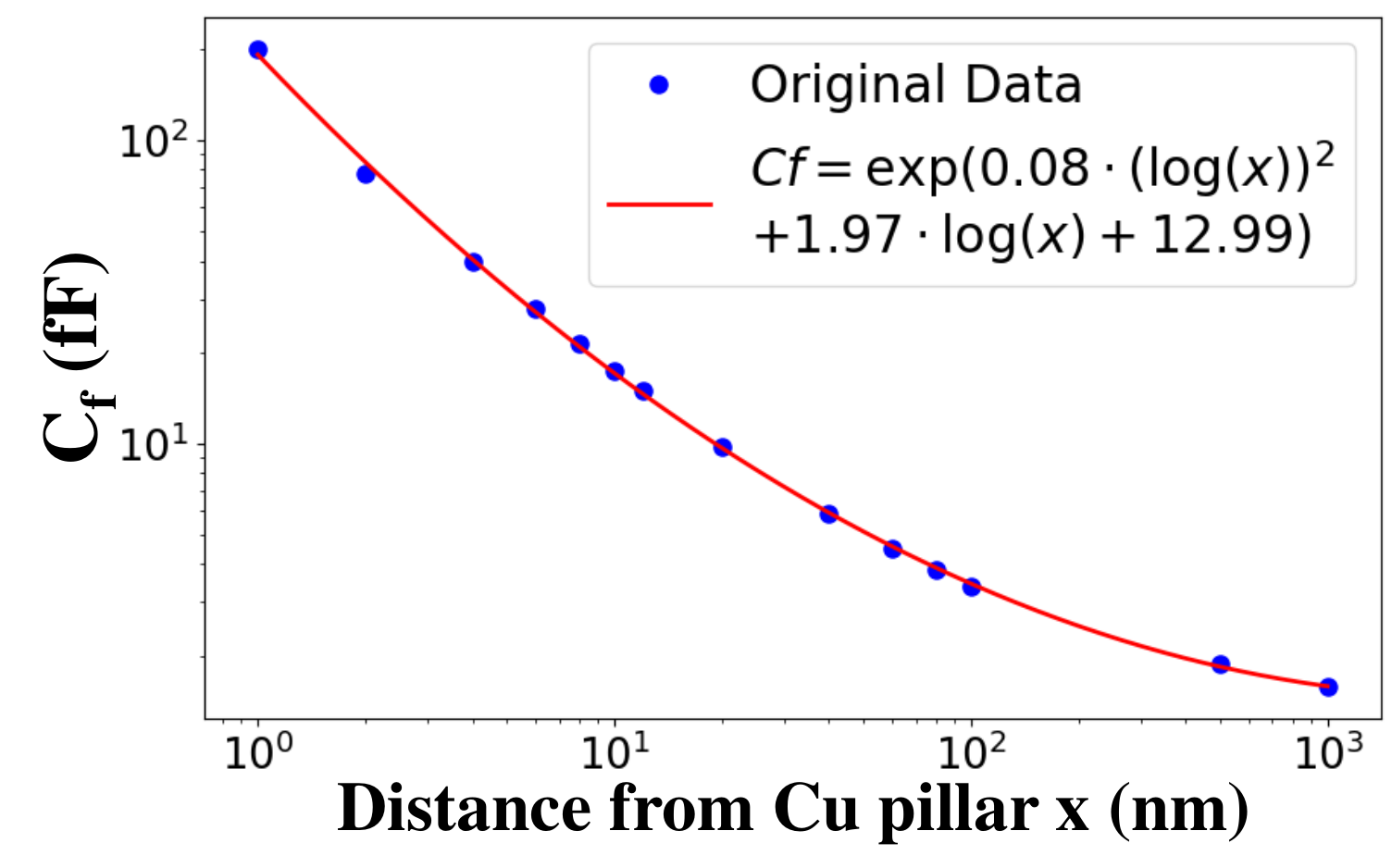}}
    \vspace{-1em}
      \\
      \subfloat[]{\includegraphics[width=0.28\textwidth]{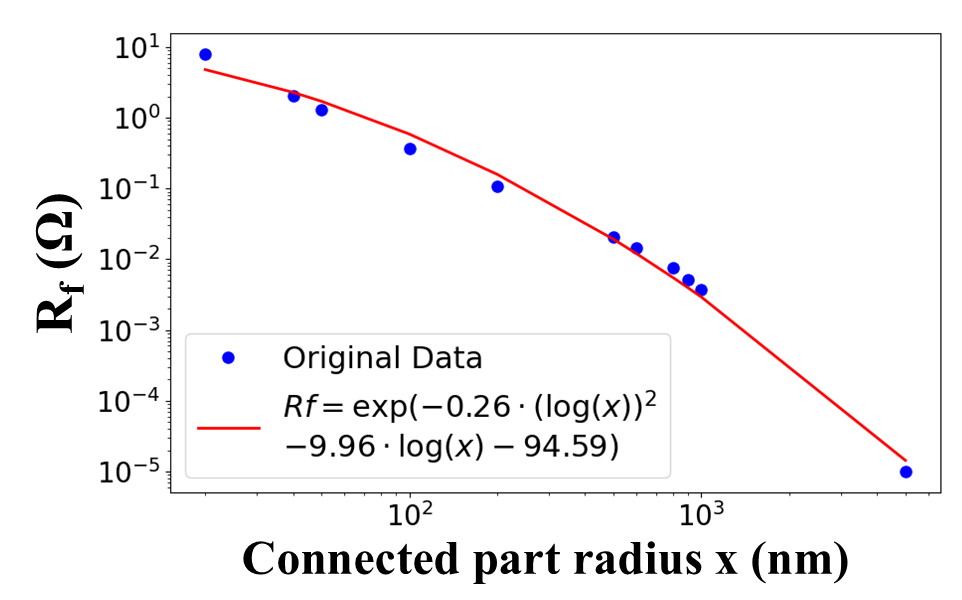}}
    \subfloat[]{\includegraphics[width=0.28\textwidth]{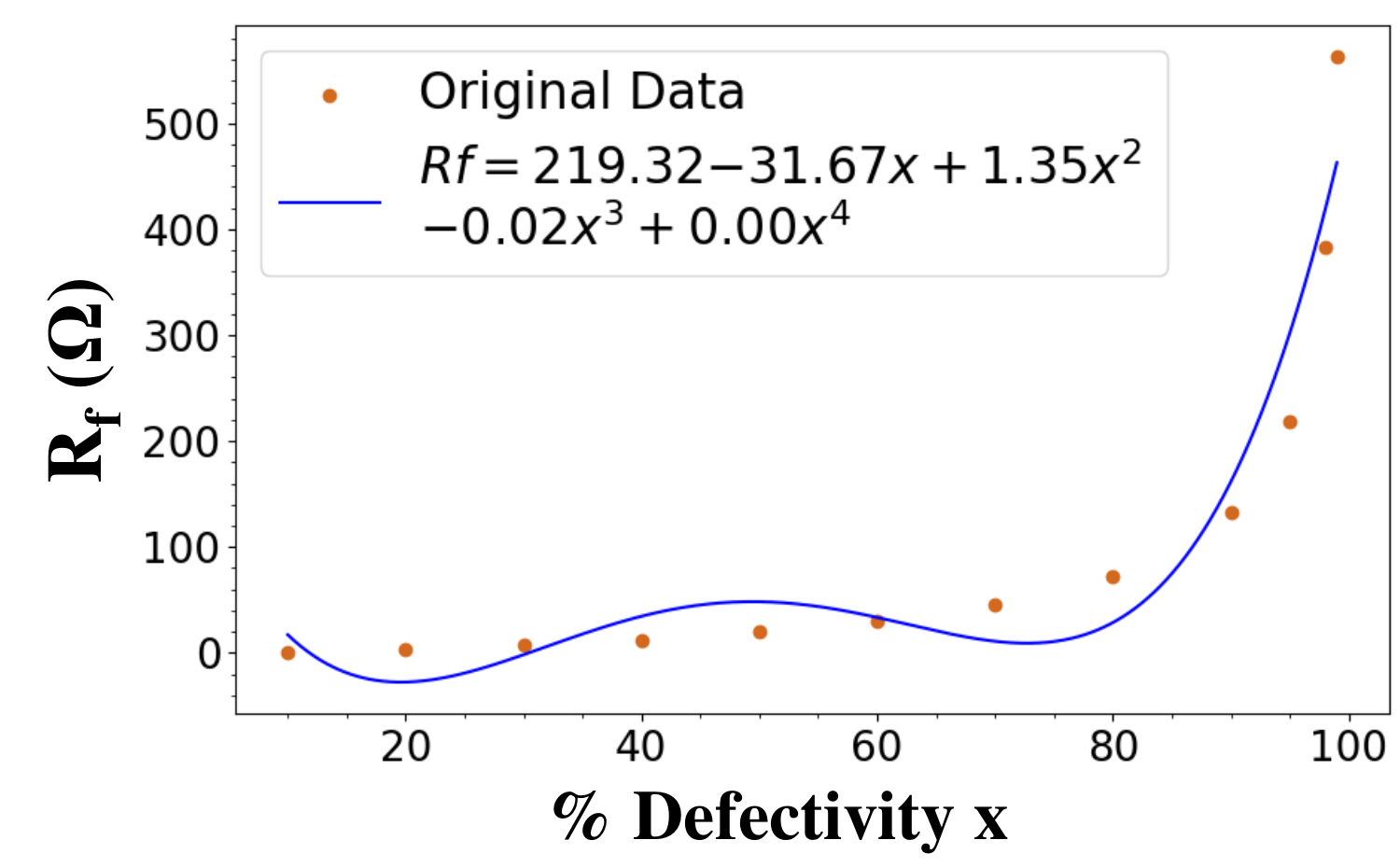}}
    \subfloat[]{\includegraphics[width=0.28\textwidth]{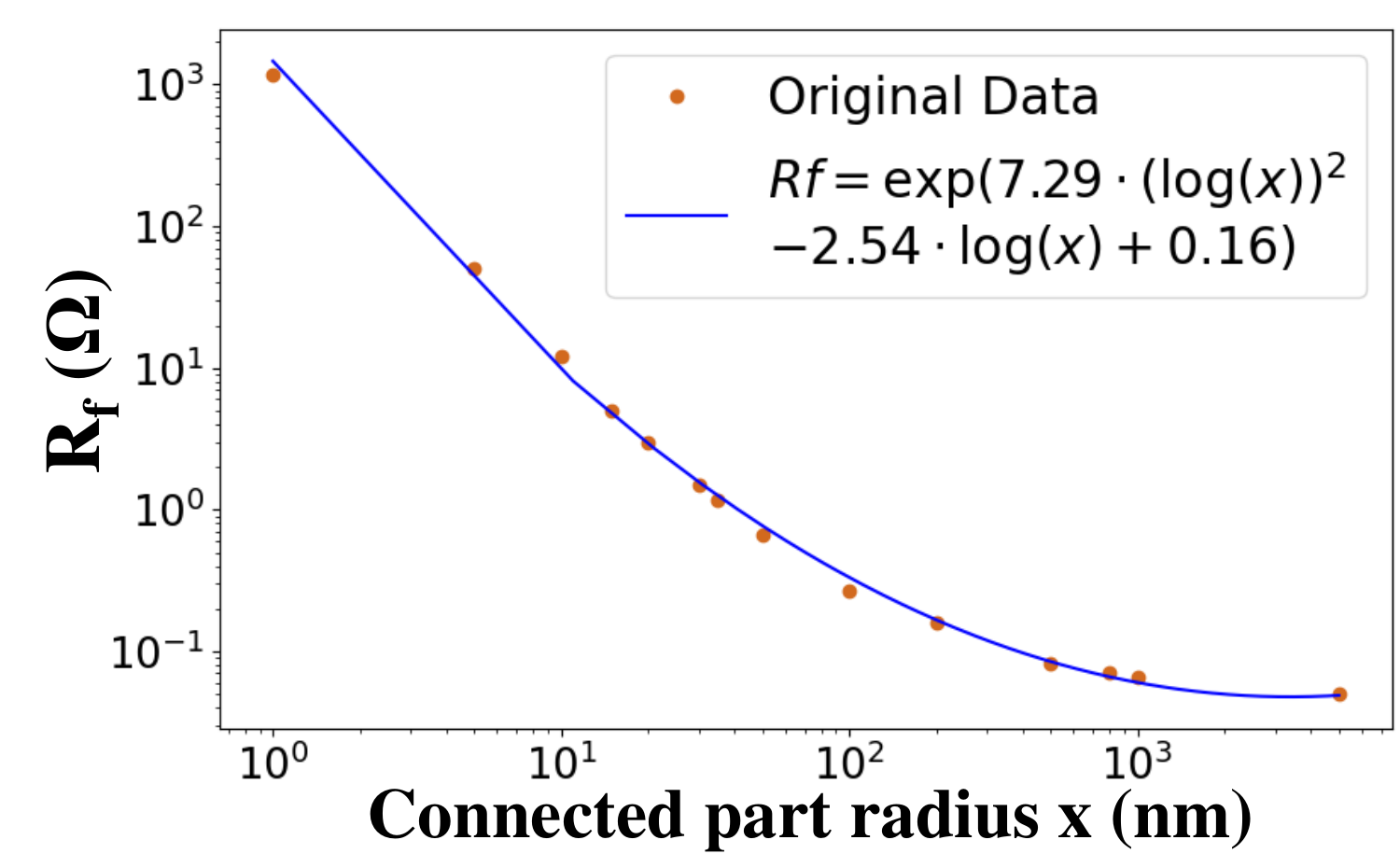}}
    
    \caption{Extracted parasitics for defects in: (a) cracked pillars; (b) shorted pillars; (c) capacitive misalignment; (d) resistive misalignment; (e) damaged RDL; (f) shorted RDL.}
    \label{fig:faulty_RLC_plot}
    \vspace{-1.4em}
\end{figure*}

To analyze defects within RDL, we created a geometry of a 10 µm section. We considered two points in this section that can be etched out in the worst-case scenario. We assume that this defect is continuous for a RDL of longer length (e.g., 100–200 µm). It is modeled by adding a resistor $R_f$ to the RDL equivalent circuit (Fig~\ref{fig:equivalent_circuits}). The extracted values correlate with the percentage of defective cross-sections of this element.


\vspace{-0.5em}
\subsection{Simulation Results} \label{sec: simulation_result}
The simulation results are summarized in Fig.~\ref{fig:faulty_RLC_plot}. For each scenario, we use Python's log-linear, polynomial, or exponential curve fitting functions and generate equations that accurately model the extracted RC data. These equations are valuable for interpolating or extrapolating defect sizes.

Recall that the elements were not completely disconnected in the case of the Cu pillar cracks; therefore, the model can only define one conductor or net. Under this condition, a capacitance solver find the capacitance $C_f$. However, since the cross-sectional area of the conductive path is small compared to the total cross-section of these elements, we extract $C_f$ from the fully broken or fractured conditions of these elements. These values can then be incorporated in the equivalent faulty circuit for the cracked case. The extracted $R_f$ values for a cracked Cu pillar are shown in Fig.~\ref{fig:faulty_RLC_plot}(a). 

The fitted model for bridging defects in the Cu pillar exhibits a steeper trend than that of the RDL, as demonstrated in  Fig.~\ref{fig:faulty_RLC_plot}(b) and Fig.~\ref{fig:faulty_RLC_plot}(f). This difference arises because the pitch of the RDL is ten times smaller than that of the Cu bumps. Therefore, when a bridging defect occurs in RDL, the impact of the defect is much more critical for signal integrity. 


Fig.~\ref{fig:faulty_RLC_plot}(c) and Fig.~\ref{fig:faulty_RLC_plot}(d) illustrate the RC parasitics for misaligned bonding of Cu pillar and RDL. An exponential relationship between the defect size and the degree of misalignment is evident, highlighting the importance of precise patterning. Fig.~\ref{fig:faulty_RLC_plot}(e) presents the extracted $R_f$ values for the defective RDL. It is important to note that these values are based on a 10 µm section of the RDL, and appropriate extrapolation is needed for other lengths.

\subsection{Fault Analysis}
\begin{figure}[!tbp]
    \centering
    \includegraphics[width=0.65\columnwidth]{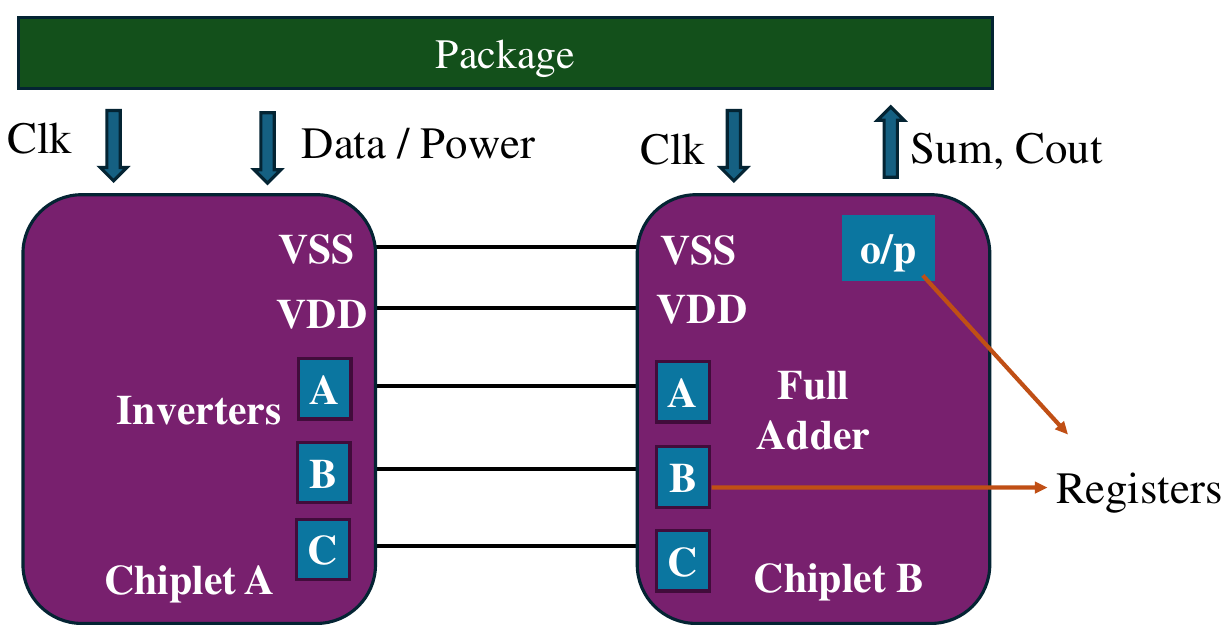}
    \caption{HSPICE setup for package-level fault simulation.}
    \label{fig:fault_sim}
\vspace{-1.9em}
\end{figure}
To analyze the functional impact of defects, we perform package-level simulations in HSPICE using OPENROAD ASAP7 PDK, as shown in Fig.~\ref{fig:fault_sim}. In this setup, Chiplet A is modeled with inverters for processing signals from the package; Chiplet B contains a full-adder circuit. All I/Os are wrapped with registers to emulate a real-world core. Using equivalent faulty circuits, we identify defect size thresholds that lead to catastrophic functional failures. These defect sizes are then mapped to the extent of structural deformities using the results presented in Section~\ref{sec: simulation_result}.

For example, capacitive misalignment or a deformation (break) in Cu bump introduces an open defect (\(C_{open}\)). Such an open in \(V_{DD}\) (\(V_{SS}\)) results in a SA-0 (SA-1) fault at chiplet B. An open in \(V_{DD}\) creates a floating node and the registers in chiplet B interpret this floating value as logic 0. Conversely, for \(V_{SS}\), the logic circuits at chiplet B establish a connection between \(V_{DD}\) and \(V_{SS}\), raising the \(V_{SS}\) potential to \(V_{DD}\), thereby causing a SA-1 fault. Interference from the signal lines on \(V_{DD}\) and \(V_{SS}\) was noticeable in the presence of open defects, but it did not cause any functional failures in the simulation. However, when an open defect occurs in signal wires (e.g., wire A), coupling effects are more prominent. For instance, if A and B are capacitively coupled, B significantly interferes with A. Beyond a critical misalignment gap ($\approx$ 20 nm), net A follows B entirely, and this behavior can be modeled as either a wired-AND (w-A) or wired-OR (w-O) fault. The threshold value of (\(C_{open}\)) increases with coupling effects, meaning that even smaller opens can trigger catastrophic faults as coupling increases.

\begin{table}[htbp]
\centering
\caption{Functional faults associated with defect size.}
\tablefontsize 
\renewcommand{\arraystretch}{1.2}
\begin{tabular}{|p{2.15cm}|p{2.8cm}|p{2.7cm}|}
\hline
\textbf{Defect Type}          & \textbf{Defect Size}            & \textbf{Functional Fault}       \\ \hline
\(V_{DD}\) Open                     & 0.1 fF $<$ \(C_{open}\) $<$ 2 $\mu$F            & \(V_{DD}\) SA-0 , O/P SA-0                  \\
\(V_{SS}\) Open                     & 0.1 fF $<$ \(C_{open}\) $<$ 2 $\mu$F             & \(V_{SS}\) SA-1, O/P SA-1                       \\ 
A open                       & \(C_{open}\) $<$ 10 fF                  & Wired-And / Wired-Or            \\ 
A shorted to \(V_{DD}\)             & \(R_{short}\) $<$ 500 $\Omega$         & A SA-1                    \\ 
A shorted to \(V_{SS}\)             & \(R_{short}\) $<$ 600 $\Omega$         & A SA-0                    \\ 
A shorted to B               & \(R_{short}\) $<$ 200 $\Omega$         & Wired-And                       \\ \hline

\end{tabular}
\label{tab:impacts_defects}
\vspace{-1.7em}
\end{table}

 Resistive shorts between signal lines and power/ground lines manifest as SA-1 (shorted to \(V_{DD}\)) or SA-0 (shorted to \(V_{SS}\)). A radial path as small as 3 nm (between the Cu bumps) or 2 nm (between the RDLs) is sufficient to cause such functional failures. Similarly, a resistive short between two signal wires (e.g., A and B) results in a w-A fault, where both signals are forced to 0 if either one is low. This behavior is observed with a radial path of 4 nm (Cu bumps) or 3 nm (RDLs). The defect sizes and their corresponding functional impacts are summarized in Table.~\ref{tab:impacts_defects}.  Any deformities smaller than these thresholds manifest as resistive defects, the analysis of which is left for future work.

\section{BIST Solution}\label{sec: bist}

In chiplet-based packages, shorts occur almost 95\% times more often than opens \cite{chakravarty2022special}. From the fault analysis in Section~\ref{sec:defect modeling}, we observe that some opens can be modeled as SAFs while others can be modeled as bridging faults due to coupling. 
Our BIST solution targets these hard SAFs and bridging faults, focusing on the high-risk areas identified from defect analysis. Testing for parametric faults and defect size estimation is left for future work. Given that FOWLP enables thousands of I/Os for each chiplet \cite{azemar2016fan}, test pattern application must be carefully optimized. For instance, when testing for a short defect, applying alternating test patterns is required for the interconnect group under test (GUT). 

The arrangement of Cu pillars (rectangular or hexagonal) affects the probability of shorts between these bumps. From \cite{chuang2024generating}, we know that in a hexagonal configuration, each bump can potentially short to any of its 12 adjacent bumps, as illustrated in Fig.~\ref{fig:bump_map}(a). By modeling this as a graph coloring problem, it has been shown that only four distinct codewords are needed to test these 12 faults. Now consider a GUT (represented as a circle in Fig.~\ref{fig:bump_map}(b)) containing 13 bumps. Within this GUT, each bump has the potential to short with multiple neighboring bumps. Each bump on the bump map generates its own ``circle" of potential shorts. These individual circles overlap because a single bump can be part of multiple circles (or GUTs), resulting in a complex, interconnected web of overlapping circles. To achieve
high fault coverage and diagnostic resolution, we must test all these overlapping GUTs simultaneously. However, this task becomes more challenging as the number of I/Os grows. This is because each bump requires a dedicated detector circuit, resulting in increased area overhead. Additionally, signals from the test pattern generator (TPG) must route to numerous bumps, causing congestion.

To address these issues, we propose dividing the bump map into several smaller blocks. Each bump within a block will receive specific test patterns. For shorts that may occur at the boundary between two blocks, we mitigate interference by setting the patterns of the blocks not currently being tested to all 0s. This ensures that any short between blocks will cause the affected bump within the active block to produce an error.
\begin{figure}
    \centering
    \subfloat[]{\includegraphics[width=0.45\columnwidth]{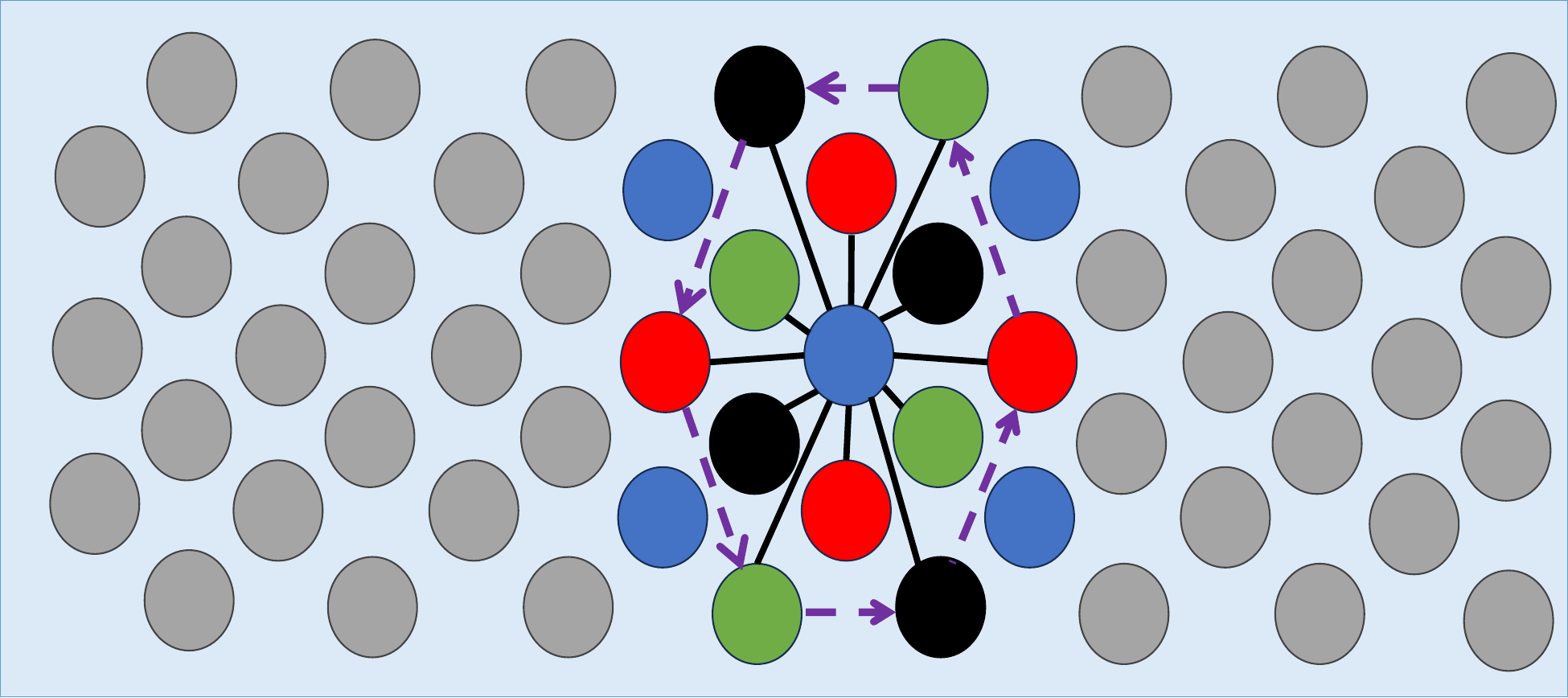}}
    \subfloat[]{\includegraphics[width=0.46\columnwidth]
    {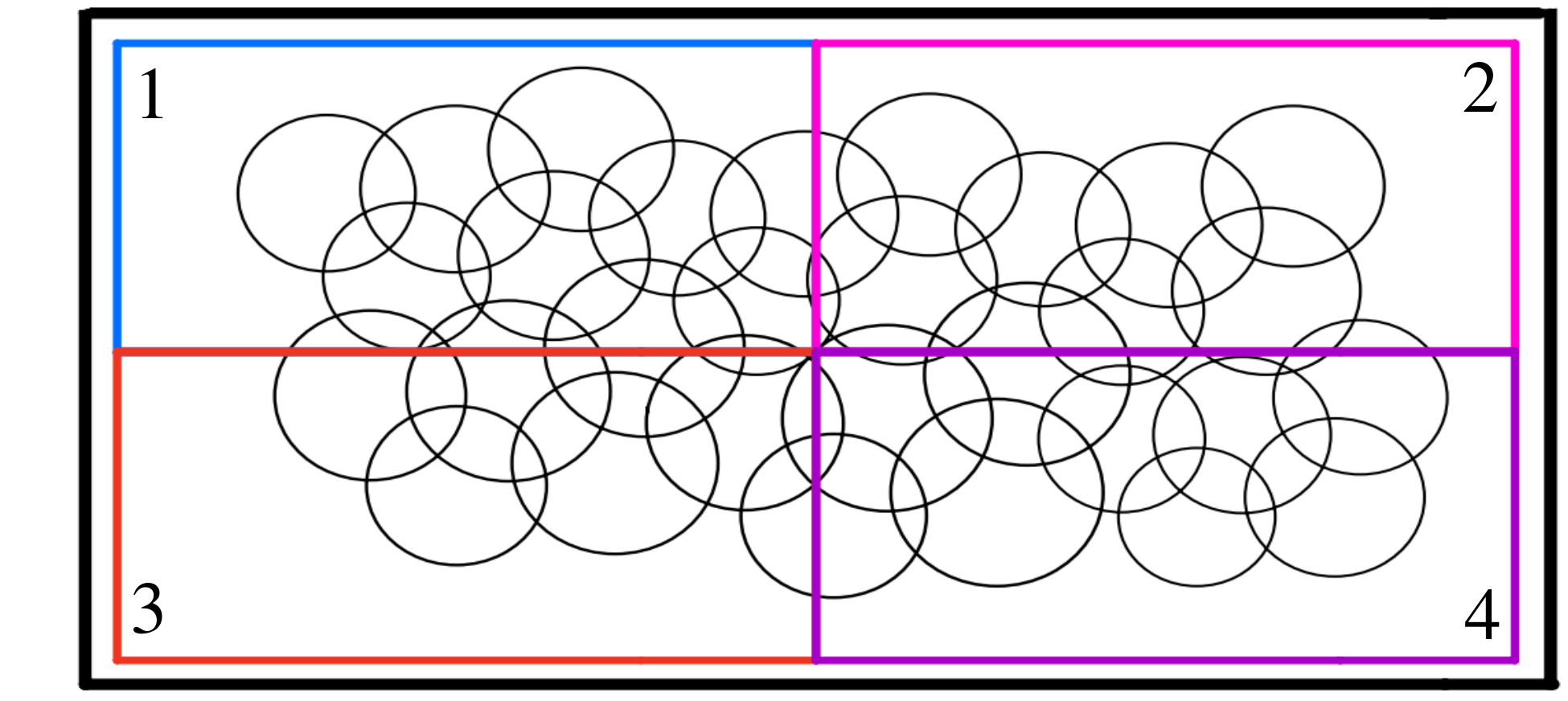}} 

    \caption{ (a) Potential bridging map (solid) for single bump   \cite{wang2024test} and illustration of a GUT (dashed) (b) Block-level division for testing.}
    \label{fig:bump_map}
     \vspace{-1.7em}
\end{figure}
\vspace{-1em}
\subsection{Detection and Diagnosis}
Since distinct test patterns need to be applied to only four colored bumps in a GUT, detecting faults within these selected bumps ensures comprehensive testing of the entire GUT. We next show that three test patterns, as illustrated in Fig.~\ref{fig:detection}(a), can detect all shorts and SAFs among these bumps. It is important to note that the input patterns for the green (blue) and black (red) bumps complement each other. Therefore, a detection circuit designed to identify faults in the green (blue) bumps can detect faults in the black (red) bumps. Fig.~\ref{fig:detection}(a) presents the detector in both the green and blue bumps. By adding an inverter, the faults in the black and red bumps can also be tested. To analyze the test results, we consider two terminals of this detection circuit: \(x\) and \(y\).  If no faults are present in the interconnects, the output will match the input signals sent from the driver side. However, if a bridging fault occurs between the bumps, the received signals will be erroneous.

\begin{figure}
    \centering
    \subfloat[]{\includegraphics[width=0.5\columnwidth]{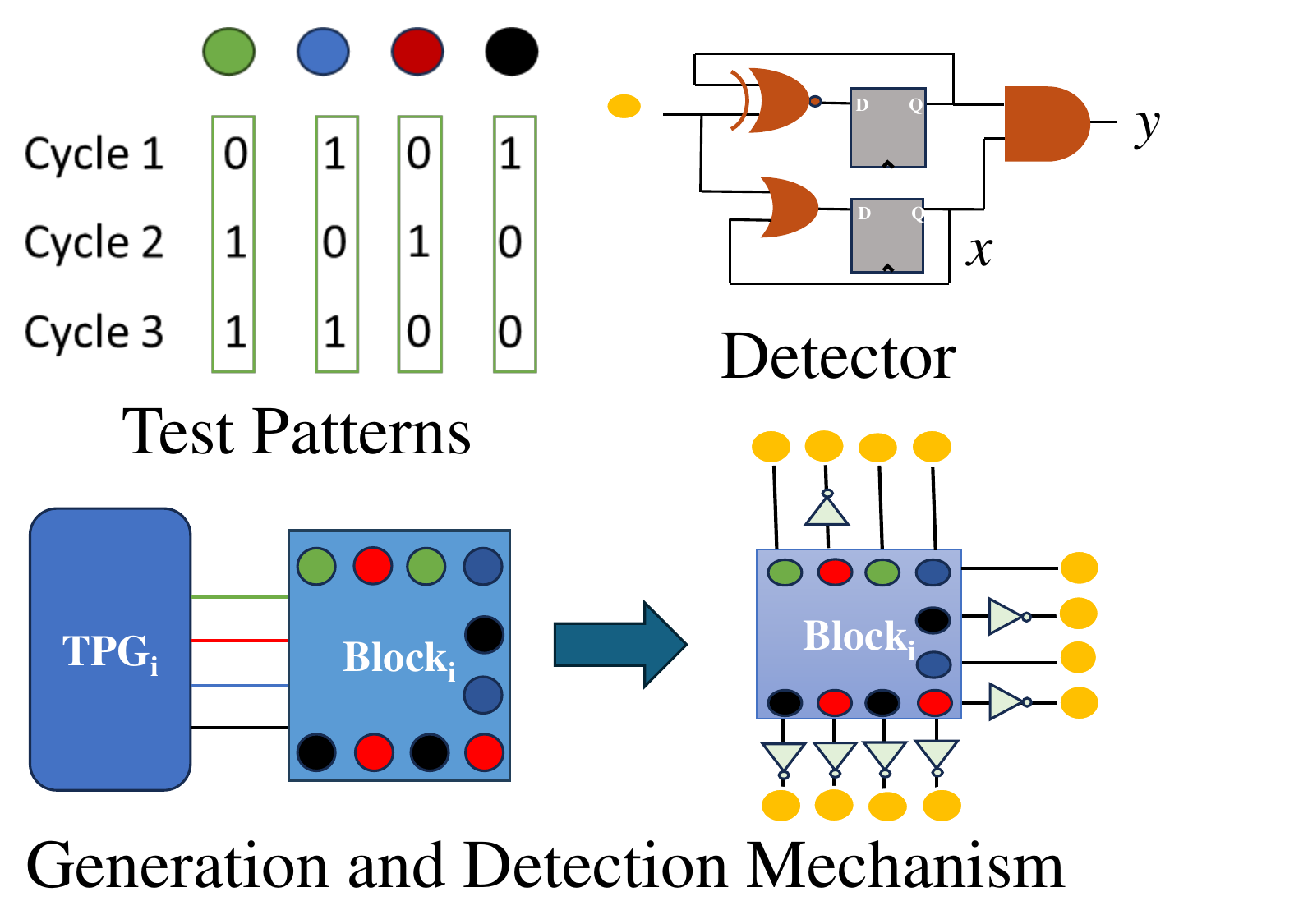}}
    \subfloat[]{\includegraphics[width=0.3\columnwidth]
   {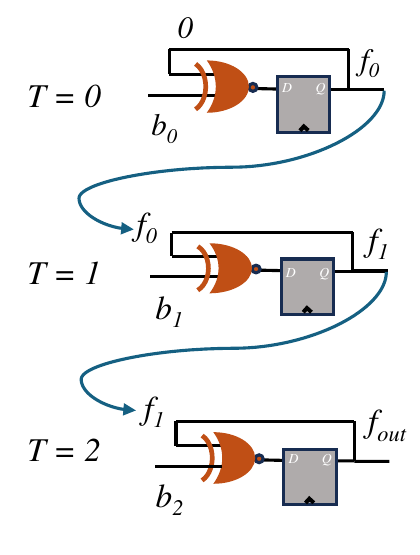}}
    \caption{ (a) Test patterns and detection circuit (b) Response analysis.}
    \label{fig:detection}
     \vspace{-1.4em}
\end{figure}

Let the incoming bits to the detection circuit at time steps \( T = 0 \), \( T = 1 \), and \( T = 2\) be \( b_0 \), \( b_1 \), and \( b_2 \), respectively. The corresponding outputs at these time steps are \( f_0 \), \( f_1 \), and \( f_{\text{out}} \); see Fig.~\ref{fig:detection}(b). It can be seen that the output of the first stage, denoted as \( f_0 \), is \( b_0' \) and \( b_2\) must equal \(f_1 \) to ensure that \( f_{out} = 1 \). Moreover, when \( b_0 = 0 \), \( b_1\) and \( b_2\) must be equal (00 or 11) whereas \( b_0 = 1\), \( b_1\) and \( b_2\)  needs to alternate (01 or 10) to yield \(f_{out} = 1\). However, in the case when all of them are 0, the OR loopback path enforces \(f_{out} = 0\). Thus, the set of inputs that makes the output 1 is \{011, 110, 101\}.

\begin{theorem}
The three test patterns shown in Fig.~\ref{fig:detection}(a) can detect all SAFs and bridging faults associated with interconnects between chiplets. In addition, these patterns can distinguish between any pair of faults with an accuracy of 95.61\%.
\end{theorem}
\begin{proof}
We enumerate all faulty responses in Table ~\ref{tab:detection}. \\
\textbf{SAFs:} A SA-0 fault in any interconnect results in an all-0 response, making both \(x\) and y equal to 0. In contrast, a SA-1 fault causes \(y\) to remain 0 while \(x\) becomes 1. Bridging defects manifest as paired faults for adjacent bumps. SA-1 faults can be differentiated from bridging defects when a single faulty response is observed.

\begin{table}[htbp]
\centering
\caption{Fault detection scenarios.}
\renewcommand{\arraystretch}{1.4}
\tablefontsize
\begin{tabular}{|p{1.3cm}|p{1.8cm}|p{2.5cm}|p{1.5cm}|}
\hline
\textbf{Faults}        & \textbf{\( b_0b_1b_2\) \newline(w-A) (w-O)} & \textbf{Response (x, y) \newline(Bump1) (Bump2)} & \textbf{Detection Outcome}                             \\ \hline
SA-0                  & 000                       & (0, 0)                  & Detected                                      \\ \hline
SA-1                  & 111                       & (1, 0)                  & Detected                                     \\ \hline
Green + Blue           & (001) \newline (111)                  & (1, 0)   (1, 0) \newline  (1, 0)   (1, 0)           & Detected in both                        \\ \hline
Green + Red            & (010) \newline (011)                     & \textbf{(1, 0)}   (1, 1)\newline  (1, 1)\textbf{ (1, 0)}          & Detected in one bump        \\ \hline
Green + Black          & (000) \newline (111)        & (0, 0)   (1, 0) \newline (1, 0)   (0, 0)        & Detected in both                              \\ \hline
Blue + Red             & (000) \newline (111)                      & (0, 0)  (1, 0) \newline(1, 0)   (0, 0)         & Detected in both                              \\ \hline
Blue + Black           & (100)  \newline (101)                     & \textbf{(1, 0) }(1, 1) \newline (1, 1)   \textbf{(1, 0)  }        & Detected in one bump      \\ \hline
Red + Black            & (000) \newline (110)                     & (1, 0)   (1, 0)   \newline(1, 0)   (1, 0)        & Detected in both                              \\ \hline
\end{tabular}
\label{tab:detection}
\end{table}

\textbf{Green Bump:} If the green bump is shorted with blue bumps, w-A and w-O generate 001 and 111, respectively. In both cases, \(y\) is 0 while  \(x\) remains 1, thus detecting the fault. When the green bump is shorted with the red bump, w-A produces 010, which matches the fault-free response for the red bump; the fault is only detected at the  \(y\) terminal of the green interconnect. Conversely, w-O produces the opposite effect. Lastly, when shorted with the black bump, the response is either all 0 or all 1, depending on the wired operation, and this fault is detected at both bumps (\(y = 0\)).

\textbf{Blue Bump: }A short between blue and red bumps can be analyzed as above. When blue is shorted to black, w-O produces a 101 signal, which is the nominal pattern for the blue interconnect. As a result, the fault is not detected at \(y\). However, it is still identified in the black interconnect. The w-A model results in the opposite outcome. 

\textbf{Red Bump:} For red and black interconnects, w-A returns all-0, which is easily detected. The w-O operation generates 110; an inverter before the detection circuit converts 110 to 001, allowing the fault to be detected.

\begin{figure}
    \centering
    \includegraphics[width=0.9\columnwidth]{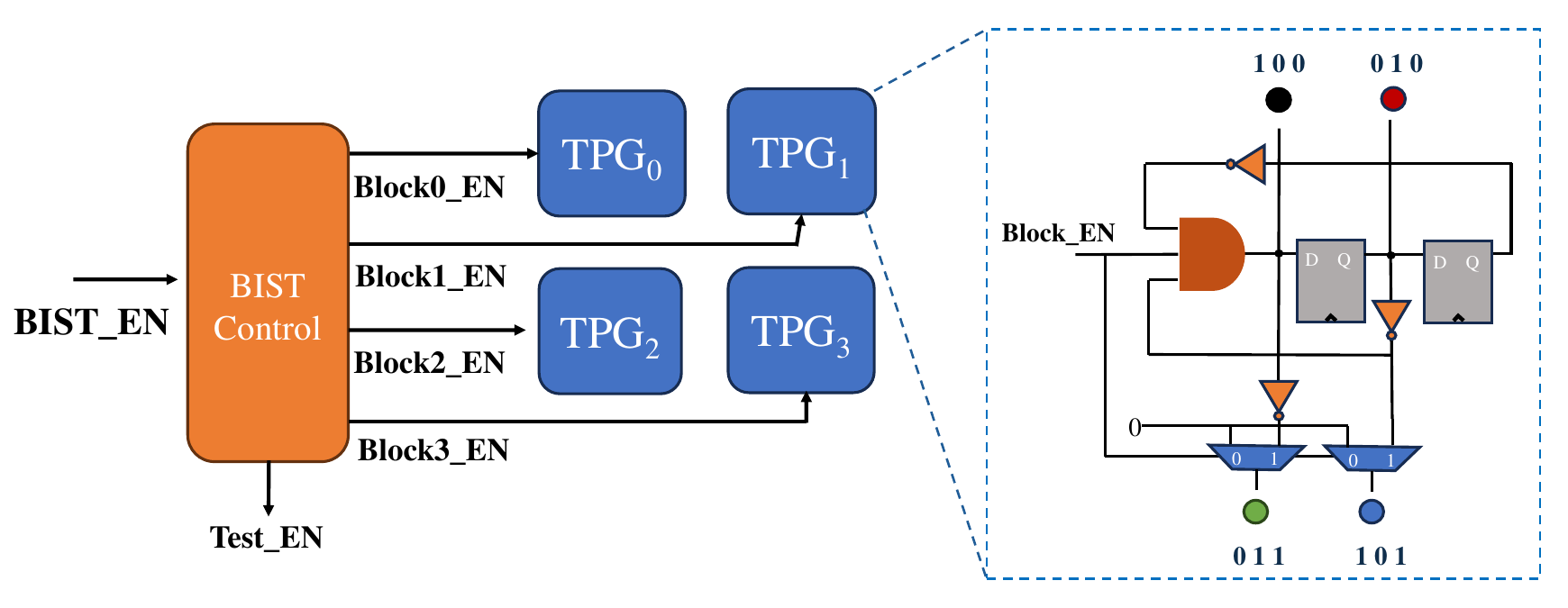}
    \caption{Design of the control circuit for test pattern application.}
    \label{fig:control}
    \vspace{-1.5em}
\end{figure}


There are two scenarios where the erroneous \(y\)-value is only observed in one bump. Take the w-A case for the green and blue bridging fault. The \((x,y\)) pair for the green bump is identical to that of a green SA-1 fault. Therefore, it is impossible to differentiate between a green SA-1 fault and a green-red bridging fault. There are three additional indistinguishable fault pairs: (blue SA-1, blue+black), (red SA-1, green+red), and (black SA-1, red+black).
The total number of faults consists of 8 SAFs and 6 bridging faults. The diagnosability \( D \) is therefore given by \( D = (\binom{14}{2} - 4)/\binom{14}{2} \approx 0.956 \).
\end{proof}
\vspace{-0.7em}

Even when a fault is not diagnosable, it is detected in at least one bump. By rerouting one of the signals to a spare bump, we can avoid discarding the package. The detected faults can be mapped to Table~\ref{tab:impacts_defects} and Fig.~\ref{fig:faulty_RLC_plot} to deduce the range of defect size and deformities in chiplet interconnects.

\subsection{Control Architecture and PPA Assessment}

The BIST control circuit is shown in Fig.~\ref{fig:control}. Upon activation of the \(BIST_{EN}\) signal, the finite state machine (FSM) transitions from the IDLE state and activates the \(Test_{EN}\) signal. The FSM starts to iterate from the IDLE state to enable all the blocks one by one; refer to Fig. ~\ref{fig:fsm}. For each block, the FSM generates a \(Block_{EN}\) signal that initiates generating patterns for all four colored bumps. All GUTs in a block can share the same TPG and receive the corresponding test patterns for each colored bump. The FSM further guarantees that blocks not under test are set to receive a zero signal by setting the other \(Block_{EN}\) signals to a low state. Two MUXes in the generator circuit ensure that when \(Block_{EN}\) is low, the bumps of that block receive a 0 signal. Furthermore, all blocks can be tested sequentially by applying the same \(BIST_{EN}\) signal at the receiving end. One \(BIST_{EN}\) signal can be used for each I/O interface to synchronize block-by-block testing, supporting communication in both directions.

\begin{figure}
    \centering
    {\includegraphics[width=0.4\textwidth]{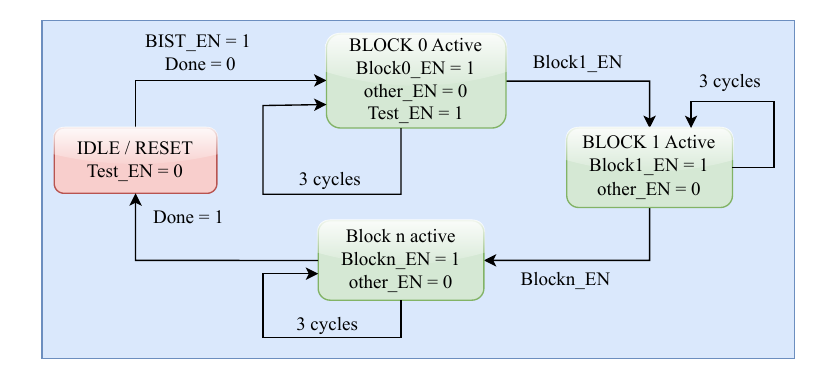}} 
    \caption{Block diagram of the FSM for BIST control.}
    \label{fig:fsm}
     \vspace{-1.6em}
\end{figure}

To evaluate PPA impact, we considered the AES and DES OpenCore benchmarks. Two copies of each core were combined by feeding the output of the first core into the input of the second, forming two hypothetical multi-chiplet packages. AES has a higher I/O count (128), while the DES is larger, with half the number of interconnects. For each design, we considered two test scenarios. In the first scenario, the interface is divided into two blocks. For AES, two test pattern generator (TPG) blocks drove test patterns to 64 I/Os (32 for the DES package), while 64 (32) detectors were used on the receiving side. The detectors can be shared between blocks to optimize area, and a MUX selects the appropriate response signals based on the specific \(Block_{EN}\) data.

\begin{table}[htbp]
{\small
\centering
\tablefontsize
\caption{PPA analysis for different configurations.}
\renewcommand{\arraystretch}{1.4}
\begin{tabular}{|c|c|c|}
\hline

\textbf{Parameter}        & \textbf{AES}            & \textbf{DES}            \\ \hline
Clock frequency  & 666.67 MHz      &  1 GHz \\ \hline
Nominal Area (µm\(^2\))      & 40231                   & 66693                   \\
2-Block BIST-ed             & 42211 \((\uparrow 4.92\%)\) & 67618 \((\uparrow 1.39\%)\)  \\
4-Block BIST-ed              & 41677 \((\uparrow 3.6\%)\)  & 67842 \((\uparrow 1.73\%)\)  \\ \hline    
Nominal Power (mW)            & 12.81                    & 57.19                   \\
2-Block BIST-ed           & 13.66 \((\uparrow 6.6\%)\)  & 57.86 \((\uparrow 1.15\%)\)  \\
4-Block BIST-ed              & 13.29 \((\uparrow 3.75\%)\) & 58.1 \((\uparrow 1.57\%)\)   \\ \hline
\end{tabular}
\label{tab:ppa}
\vspace{-0.6em}
}
\end{table}


In the second scenario, the I/Os were divided into four blocks, each with 32 (16) I/Os and 32 (16) detectors. This setup required the deployment of four TPG blocks. For both test cases, we synthesized the designs using the 45 nm Nangate standard cell library; the results are shown in Table~\ref{tab:ppa}.



Transitioning from a 2-block to a 4-block configuration for AES leads to a reduction in both power and area overhead. However, note that while the 4-block configuration is more PPA-efficient, it results in increased test time (due to sequential testing) compared to the 2-block configuration. The DES case study involves fewer I/Os. Although the detector circuits are fewer, more TPG blocks leads to more control circuits, and the combinational circuit requirement becomes dominant, increasing area overhead. The total area overhead is proportional to the number of I/Os per interface; adding more blocks becomes beneficial as the number of I/Os increases. 

\section{Conclusion}\label{sec:conclusion}
We have presented a defect analysis framework for FOWLP and demonstrated how these defects can be mapped to their corresponding faulty circuits. Through fault simulation on a test package in HSPICE, we illustrated what level of deformities can lead to catastrophic functional failures. Additionally, we introduced a BIST architecture that achieves full fault coverage using only three test patterns while also providing fault diagnosis with over 95\% accuracy. Results for two benchmarks highlight how design decisions can be optimized with respect to PPA and test time constraints.

\section*{Acknowledgment}
This research was supported in part by the National Science Foundation in part by CHIMES.

\bibliographystyle{IEEEtran}
\bibliography{reference}

\end{document}